\documentclass{llncs}

\usepackage{epsfig}
\usepackage{subfigure}
\usepackage{graphics,graphicx}
\usepackage[textsize=footnotesize]{todonotes}

\newtheorem{statement}{Statement}
\newcommand{\gate}{\textrm{gate}}
\newcommand{\vis}{\textrm{vis}}

\newcommand{\blind}{\textrm{blind}}
\newcommand{\R}{\textrm{R}}

\title{Simultaneous Straight-line Drawing of a Planar Graph and Its Rectangular Dual }

\author{
Tamara  Mchedlidze}
\institute{Faculty of Informatics, Karlsruhe
Institute of Technology (KIT), Germany \texttt{mched@iti.uka.de}}

\begin{document}
\maketitle

\begin{abstract}
A natural way to represent on the plane both a planar graph and its dual is to follow the definition of the dual, thus, to place vertices inside their corresponding primal faces, and to draw the dual edges so that they only cross their corresponding primal edges. The problem of constructing such drawings has a long tradition when the drawings of both primal and dual are required to be straight-line. We consider the same problem for a planar graph and its rectangular dual.  We show that the rectangular dual can be resized to host a planar straight-line drawing of its primal.  
\end{abstract}

\section{Introduction}

A \emph{planar drawing} of a planar graph is its representation  on the plane such that its vertices  are mapped to distinct points and its edges to non-intersecting simple Jordan curves. A drawing is called \emph{straight-line} if each edge is represented by a line segment.  It is well-known that each planar graph admits a planar straight line drawing~\cite{Fary48},  even with a quadratic area~\cite{DFPP90,Schnyd90}.
A planar drawing $\Gamma$ partitions the plane into topologically connected regions called \emph{faces}, the unbounded face is called \emph{external} and the remaining are called \emph{internal} faces. The edges  that bound the external face are also called \emph{external}, the remaining edges are \emph{internal}.  A \emph{planar embedding} of a planar digraph $G$ is an equivalence class of planar drawings that induce the same clockwise cyclic ordering of edges around each vertex and that have the same external face.

An alternative way to represent a planar graph $G$  is to draw the vertices as geometric shapes so that, two shapes touch\footnote{We say that two shapes touch if they have a common interval of a positive length.} if and only if the corresponding vertices of $G$ are adjacent. Such type of representation is called \emph{contact representation}. Different kinds of shapes for contact representations of planar graph have been considered(eg.,~\cite{AlamBFKK12,FraysseixMR94,DuncanGHKK12,GoncalvesLP12}). One of the most  simple  for the visual perception is a contact representation with rectangles. 
A contact representation with rectangles is called  \emph{rectangular subdivision}, if it forms a partition of a rectangle into a set of smaller non-intersecting rectangles such that no four of them meet at the same point. Such contact representation of a planar graph $G$ is known as a \emph{rectangular dual} of $G$, and is denoted by $D$. Figure~\ref{fig:rect_dual} shows a planar graph and its rectangular dual.  Graph $G$  is referred to as \emph{primal} of $D$. Unfortunately, not all planar graphs admit a rectangular dual. In particular, a planar graph $G$ has a rectangular dual $D$ with four rectangles on the boundary if and only if every internal face of $G$ is a triangle, the external face is a quadrangle and there is  no separating triangles in $G$~(see eg.,~\cite[Theorem~2.1]{He99}). The condition that $D$ is bounded by four rectangles can be relaxed~\cite{He93}.

\begin{figure}[htb]
\centering
\includegraphics[scale=0.7]{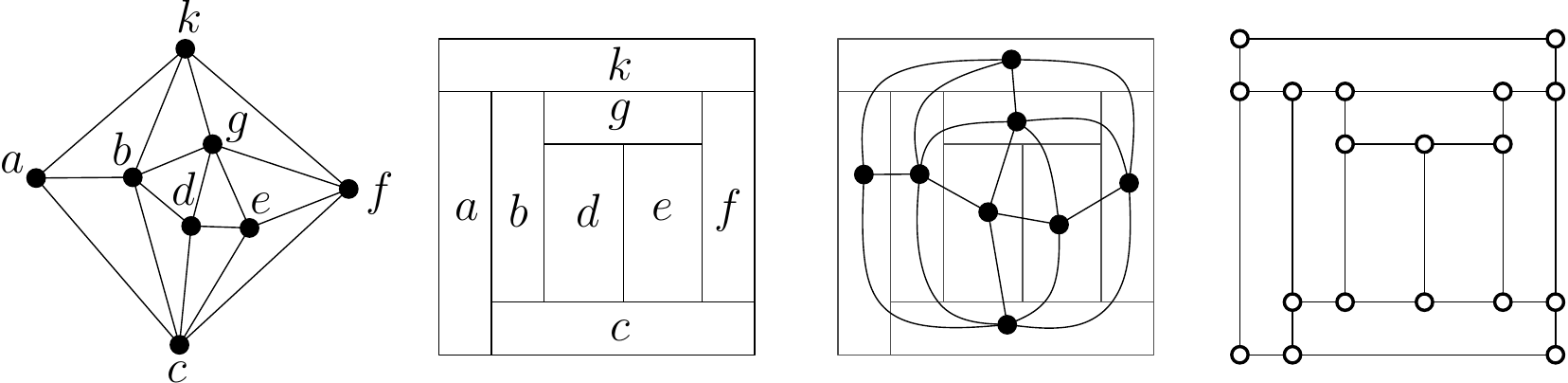}
\caption{From left to right: A graph $G$, its rectangular dual $D$, a simultaneous planar-rectangular dual drawing of $G$ and $D$, graph $G'$ constructed from $D$. }
\label{fig:rect_dual}
\end{figure}

A natural way to simultaneously represent both a planar graph $G$ and its rectangular dual $D$ is to draw a vertex of $G$ inside its corresponding rectangle  and to draw each edge $(u,v)$ as a curve crossing only the common segment of the rectangles representing $u$ and $v$ (see Figure~\ref{fig:rect_dual}, second from the right). To simplify the visual complexity of such a representation one may ask for a straight-line drawing of $G$.  Such a drawing is called \emph{straight-line simultaneous  drawing} of a graph $G$ and  its rectangular dual $D$.    It is not surprising that if the drawing of a rectangular dual $D$ is fixed, then it might not be possible to position the vertices of $G$  in order to obtain a straight-line simultaneous drawing of $G$ and $D$~\cite[Lemma~1]{AlamKKM14}. It is also known that the corresponding  decision problem is $\cal NP$-hard~\cite[Theorem~1]{AlamKKM14}. In this paper, we show that if the rectangular dual $D$ is allowed to be resized, i.e. we are allowed to change the sizes of the rectangles (without further changing the structure) we can achieve a straight-line simultaneous drawing of the primal graph $G$ and of this resized rectangular dual. 

The problem of straight-line simultaneous drawing of a graph and its rectangular dual finds application in visualization of clustered graphs~\cite{AlamKKM14,HuGK10}. A rectangular subdivision can be seen as a simplification of a map. Assume that each region of this map contains some elements related to each other (a cluster, and intra-cluster edges) and to  the elements of adjacent regions (inter-cluster edges). Some possible readability requirements for a visualization of this network together with the map are that the entire network is drown in a planar fashion, the intra-cluster edges must lie completely inside the corresponding region and the inter-cluster edges must cross only the common segment of the regions where their end-points lie.   A simple approach to construct such a visualization is to contract each cluster to a vertex, which results in the graph primal to the given rectangular subdivision. Then, construct a straight-line simultaneous drawing of the primal and the rectangular dual and, finally, uncontract the clusters. This approach is described in detail in~\cite[Theorem~2]{AlamKKM14}.   

Allow us, until the end of this section, to reverse the roles of $G$ and $D$. In particular, consider the graph $G'$, the vertices of which are the corners of the rectangles of $D$ and the edges of which are the parts of the sides of the rectangles connecting the vertices (see Figure~\ref{fig:rect_dual}, right). The edges of this graph are represented either by horizontal or by vertical segments, and each face, including the external, forms a rectangle, i.e. $D$ implies  a so-called \emph{rectangular drawing} of $G'$~\cite{gdhandbook13rect}.  Then, graph $G$ becomes a \emph{week dual}\footnote{The  \emph{dual graph} $G^\star$ of a planar graph $G$ with a fixed planar drawing is formed by placing a vertex inside each face of $G$, and connecting  vertices of $G^\star$ whose corresponding faces in $G$ are adjacent. A \emph{week dual} of the graph $G$ results 
by removing from $G^\star$ the vertex representing the external face of $G$.
Graph $G$ is  called the \emph{primal} of $G^\star$.
} of this new graph $G'$. 
When asking for a straight-line simultaneous drawing of $G$ and $D$ we are actually asking for the  straight-line simultaneous drawing of the primal $G'$ and its dual $G$. This point of view helps us to summarize the related work in the next paragraphs. 

Drawings of the primal and dual graphs so that each vertex of the dual is placed inside the corresponding face of the primal and each dual edge crosses only the corresponding primal edge will be referred to as  \emph{simultaneous planar-dual drawing}. Such a drawing is called \emph{straight-line} if both graphs are drawn straight-line.  It is immediately clear that in case of a non-week dual   a straight-line simultaneous planar-dual drawing does not exist and at least one edge (dual to an external edge)  need to have a bend. To avoid this special case from now on we only consider the week dual graph, without further mentioning  it.

Already back in 1963 Tutte~\cite{Tut63} considered the problem of constructing a straight-line simultaneous  planar-dual drawing and showed that it exists when the primal graph is triconnected. The drawing constructed by Tutte's algorithm may have exponentially large area. Only four decades later,  Erten and Kobourov~\cite{ErtenK05} provided a linear-time algorithm to construct a straight-line simultaneous planar-dual drawing on a grid of size $(2n-2) \times (2n-2)$ for the same family of graphs. Later, Zhang and He~\cite{ZH06} improved this result to a grid of size $(n-1) \times n$.

Observe that talking about straight-line simultaneous planar-dual drawing, we ask for the construction of the drawings of both the primal and the dual. A stricter variation of the straight-line simultaneous  planar-dual drawings was considered by Bern and Gilbert~\cite{BernG92}, where the drawing of the primal graph is fixed and one only need to determine the positions of dual vertices. Notice that here it is also required that each dual edge crosses only the corresponding primal. Bern and Gilbert observed that the problem is easy if all faces of the primal are triangles, thus the dual vertices can be placed at the meeting points of angle bisectors.
They presented a linear time algorithm to construct the drawing of the dual in case where convex quadrilateral faces are also present. They showed that a straight-line drawing of the dual does not always exist if non-convex quadrilaterals are present. Finally, they proved that the decision problem is $\cal NP$-hard for five-sided convex faces. 

Specially convenient for a discretization method~\cite{Baker88}  are straight-line simultaneous planar-dual drawings with the additional requirement that the primal and dual edges cross at right-angle.  Such drawing always exists if each internal face of the primal graph  is a non-obtuse triangle. In particular, the dual graph can be drawn by joining perpendicular bisectors of the edges~\cite{Bern92}. The requirement of right-angle crossing in  straight-line simultaneous planar-dual drawings have also been studied in~\cite{ArgyriouBKS13,Brightwell93,Mohar97}, see~\cite{gdhandbook13} for an overview.

In this paper we show that a given rectangular dual $D$ of a planar graph $G$ can be resized so that $G$ and this resized rectangular dual obtain a simultaneous straight-line drawing. Proof of this statement is the subject of Section~\ref{sec:main}. Before proving this result we introduce the necessary definitions in Section~\ref{sec:definitions} and give preliminary observations in Section~\ref{sec:visibility}.

\section{Definitions and useful facts}
\label{sec:definitions}

\paragraph{\bf $st$-digraph and its parts} Consider now a directed graph $G$, \emph{digraph}, for short. A \emph{source} (resp. a \emph{sink}) of  $G$ is a vertex with only outgoing  (resp. incoming) edges. An \emph{$st$-digraph} is an acyclic digraph with exactly one source $s$ and exactly one sink $t$. A \emph{planar $st$-digraph} is an $st$-digraph that is planar and provided with a planar embedding such that vertices $s$ and $t$ lie on the boundary of the external face(Figure~\ref{fig:st-digraph}, left).  It is common to visualize planar $st$-digraphs in an \emph{upward} fashion, i.e. with edges represented by curves monotonically increasing in upward direction. 

It is not hard to see (refer also to~\cite[Lemma~4.1]{diBattistaThebook}) that a face $f$ of a planar $st$-digraph $G$ is bounded by two directed paths meeting only at the source and at the sink of $f$ (see Figure~\ref{fig:st-digraph}, left).  If we imagine $G$ being embedded upward we can characterize these paths as the \emph{left} and the \emph{right  boundaries} of $f$. Let $e$ be an edge of $G$, the face of $G$ lying to the left (resp. right) of $e$ is called \emph{left} (resp. \emph{right}) face of $e$ (see again Figure~\ref{fig:st-digraph}, left).

\begin{figure}[htb]
\centering
\includegraphics[scale=0.7]{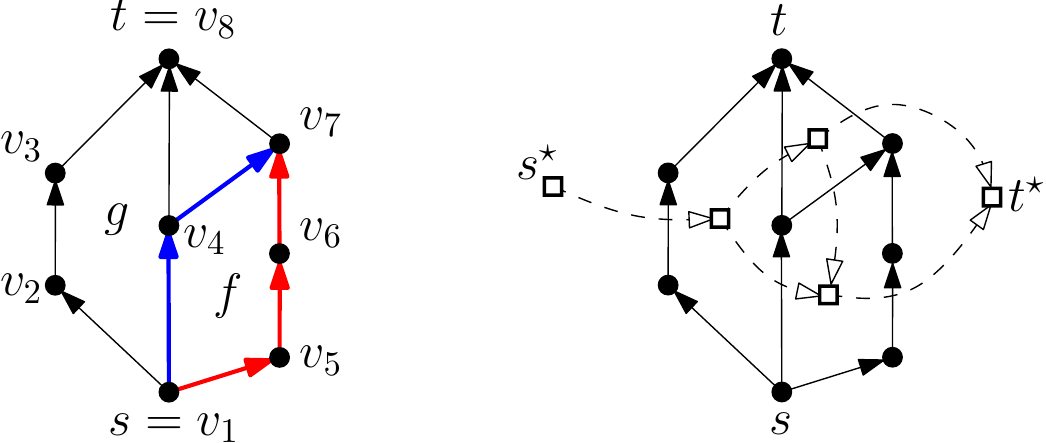}
\caption{Left: A planar $st$-digraph $G$. The blue (resp. red) path comprises the left (resp. right) boundary of face $f$. Face $g$ (resp. $f$) is the left (resp. right) face of the edge $(s,v_4)$. Indicies of the vertices are given accordingly to a topological ordering of the vertices. Right: The dual $G^\star$ of the st-digraph $G$ is depicted by dashed edges, the multiple edges are merged. }
\label{fig:st-digraph}
\end{figure}

\paragraph{\bf Topological ordering and dual digraph} A \emph{topological ordering} of a digraph $G=(V,E)$  is a 1-1 function $\rho:V\rightarrow \{1\dots |V|\}$ such that for every edge $(u, v)$ we have $\rho(u)<\rho(v)$ (Figure~\ref{fig:st-digraph}, left).

We define the \emph{dual} digraph $G^\star$ of a planar $st$-digraph $G$ as follows. The vertex set of $G^\star$ is the set of internal faces of $G$ plus the two vertices, $s^\star$ and $t^\star$, for the external face of $G$, where $s^\star$ is for its left  and $t^\star$ for its right boundary (Figure~\ref{fig:st-digraph}, right). For every edge $e\neq (s,t)$ of $G$, $G^\star$ has an edge $e^\star=(f,g)$ where $f$ and $g$ are the left and the right faces of $e$, respectively.  Digraph $G^\star$ is generally a multigraph, but in this work we merge the multiple edges to one. Digraph $G^\star$ is an $st$-digraph with the source $s^\star$ the sink $t^\star$~\cite{diBattistaThebook}.

\paragraph{\bf Rectangual dual}  Let $G$ be a graph. A \emph{rectangular dual} $D$ of $G$ is a rectangular subdivision  $\cal R$ and a one-to-one correspondence between the vertices of $G$ and the rectangles of $\cal R$ such that two vertices are adjacent in $G$ if and only if their corresponding rectangles share a common boundary.  If two vertices are adjacent we say that the corresponding rectangles are also \emph{adjacent}. For the sake of simplicity we use the same notation for the vertices of $G$ and for the rectangles of $D$.

\paragraph{\bf Simultanous drawing of a planar graph and its rectangular dual}
Let $G$ be a graph admitting a rectangular dual $D$. We say that $G$ and $D$ have a \emph{straight-line simultaneous drawing}, if we can place each vertex of $G$ inside its corresponding rectangle of $D$ such that if the edges of $G$ are drawn straight-line, the resulting drawing of $G$ is planar and each edge $(u,v)$ crosses $D$ only at a single point contained in the common boundary of the rectangles representing $u$ and $v$.

\paragraph{\bf Notation and operations for rectangles and rectangular dual} Let $u$ be a rectangle on the plane with edges parallel to coordinate axes. We denote by $x_1(u),~x_2(u),~y_1(u),~y_2(u)$ the x- and y-coordinates of the corners of $u$, where $x_1(u)<x_2(u)$ and $y_1(u)<y_2(u)$. We denote by $R(u)$ the rightmost segment of $u$. Let $v$ be a different rectangle adjacent to $u$. We
denote by $[u,v]$ the maximal common segment of $u$ and $v$.  If segment $[u,v]$ is vertical, we denote by $y_1[u,v]$,  $y_2[u,v]$ the y-coordinates of its end-points. We say that $u$ and $v$ have \emph{vertical} (resp. \emph{horizontal}) adjacency if the segment $[u,v]$ is vertical (resp. horizontal).

Given a rectangle $u$ on the plane, we define as \emph{stretch} of $u$ as the increase of $x_2(u)$.  Let $D$ be a rectangular dual of a planar graph $G$. We define a \emph{scaling} of $D$ to be a rectangular dual of $G$ that results from resizing of some of the rectangles of $D$. Observe that scaling does not change the type of the adjacency(vertical or horizontal) of two adjacent rectangles of $D$. 

\paragraph{\bf Visibility} Our target is to construct a straight-line simultaneous drawing of a planar graph $G$ and its rectanglular dual $D$.  When we place vertex $u$ inside its corresponding rectangle in $D$ we denote by $x(u)$, $y(u)$ the coordinates of $u$.  Our main requirement is that an edge $(u,v)$ of $G$ crosses the boundaries of the rectangles of $D$ only at a single point and particularly it only crosses the segment $[u,v]$.   To work with this requirement, we need several additional definitions which reflect the notion of visibility of a vertex inside a rectangle to an adjacent rectangle.  Let $u$ and $v$ be two adjacent vertices of $G$, and assume that the position of vertex $u$ in its corresponding rectangle of $D$ is fixed. The \emph{visibility region} of $u$ inside $v$, denoted by $\vis(u,v)$, is the region delimited by the boundary of the rectangle $v$ and the lines through $u$ and the two end-points of the segment $[u,v]$, see for example Figure~\ref{fig:hor_adj} and Figure~\ref{fig:diverg_adj}.  If $u$ and $v$ have horizontal adjacency such that $u$ is above $v$ and $x_1(u)<x_1(v)$, $x_2(u)=x_2(v)$ then $\vis(u,v)$ contains the topmost segment of $R(v)$ (Figure~\ref{fig:hor_adj}). In case $u$ and $v$ have vertical adjacency we distinguish two types of visibility region as follows. We say that the visibility region $\vis(u,v)$ is \emph{diverging} and that $u$ is a \emph{diverging neighbor} of $v$  if the two lines through $u$ delimiting $\vis(u,v)$ have slopes of a different sing. Figure~\ref{fig:diverg_adj}, depicts several cases of diverging visibility regions. Figure~\ref{fig:nondiverg_adj}, top left, depicts non-diverging visibility regions.   

Assume again that $u$ and  $v$ have a vertical adjacency and $u$ is to the left of $v$. In case $u$ is a non-diverging neighbor of $v$ we have that either $(1)$ $y_1(v)<y_1(u)<y_2(v)<y(u)<y_2(u)$ or $(2)$ $y_1(u)<y(u)<y_1(v)<y_2(u)<y_2(v)$. See Figure~\ref{fig:nondiverg_adj}, top left for the illustration of the case $(1)$. Consider the set  $R(v)\setminus \vis(u,v)$, it is  non-empty and generally contains two segments; one containing the topmost point of $R(v)$ and one containing the botommost point of $R(v)$. We denote by $\blind(u,v)$  the segment of $R(v)\setminus \vis(u,v)$ which contains the topmost (resp. bottommost) point of $R(v)$ for the case $(1)$ (resp. $(2)$) (Figure~\ref{fig:nondiverg_adj}, top left). 

During the construction of the straight-line simultaneous drawing of $G$ and its rectangular dual $D$ we mostly place vertices close to the right boundary of the rectangles. To formalize this we use the following notation. Let $u$ be a rectangle, we denote by $\gate(u)$  a proper sub-interval of $R(u)$ not containing both end-points of $R(u)$.

\paragraph{\bf Regular edge labeling}
Let $G$ be a planar embedded planar graph with no separating triangle, with exactly four vertices on the external face and each internal face being a triangle. It is known that such a graph has a rectangular dual  $D$(see e.g.~\cite{He93}). Two adjacent rectangles of $D$ have either vertical or horizontal adjacency. This fact is mirrored by  so-called \emph{regular edge labeling} (REL, for short)~\cite{He93}, defined for graph $G$, which is also known as \emph{transversal structure}~\cite{Fusy09}. It is formally defined as follows. A REL of $G$ is a partition and orientation of its interior edges resulting in two disjoint sets of arcs $E^R$ and $E^B$, so that:
\begin{itemize}
\item For each internal vertex $u$ the edges incident two $u$ appear in the counterclockwise order as follows: edges in $E^R$ outgoing from $u$, edges in $E^B$ incomming to $u$, edges in $E^R$ incomming to $u$, and edges in $E^B$ outgoing from $u$; moreover, none of these four sets of edges is empty;
\item Four outer vertices of $G$ are named $v_N$, $v_S$, $v_W$, and $v_E$. Moreover, the internal edges incident to $v_S$ (resp. $v_N$) are all in $E^R$ and are outgoing from $v_S$ (resp. incomming to $v_N$). Also, the internal edges incident to $v_W$ (resp. $v_E$) are all in $E^B$ and are outgoing from $v_W$ (resp. incomming to  $v_E$).
\end{itemize}
It is known that every planar graph without separating triangles, with exactly four vertices on the outer face and each internal face being a triangle has a REL~\cite[Theorem~2.2]{He93}. Such a REL is used as a tool for constructing a rectangular dual of $G$. Let $G^R$ (resp. $G^B$) be the directed subgraph of $G$ induced by the edges in $E^R$ (resp. $E^B$) and the four exterior edges directed such that $v_S$ (resp. $v_W$) is a source of $G^R$ (reps. $G^B$) and $v_N$ (resp. $v_E$) is  a sink of $G^R$ (reps. $G^B$).  We will heavily rely on the fact that $G^R$ is a planar $st$-digraph with source $v_S$  and sink $v_N$~\cite[Lemma 2.3]{He93}. We use \emph{red} and \emph{blue} colors to distinguish edges in $E^R$ and $E^B$.

Observe that given a rectangular dual  $D$ of $G$, one can construct a REL by: $(a)$  placing internal edges of $G$ depicting horizontal adjacency to $E^R$  and orienting them from bottom to top, and $(b)$ placing internal edges depicting vertical adjacency to $E^B$ and orienting them from left to right. This REL is said to be \emph{defined by} the rectangular dual $D$. The reverse is also true, given a REL of $G$ one can construct a rectangular dual $D$ such that the blue edges specify vertical and the red edges horizontal adjacency~\cite[Theorem~4.3]{He93}. This $D$ is said to be \emph{consistent} with the given REL.

\section{On visibility between two adjacent rectangles}
\label{sec:visibility}

The following statement  is illustrated in Figure~\ref{fig:hor_adj}.
\begin{statement}
\label{stat:horizontal}
Let $u$ and $v$ be two horizontally adjacent rectangles, such that $u$ is above $v$,  $x_1(u)<x_1(v)$ and $x_2(u)=x_2(v)$.  There exists $X \geq x_2(v)$ such that, $\forall x\geq X$,  if we set $x_2(v) =x_2(u)= x$  then $\gate(v) \subset \vis(u,v)$. 
\end{statement} 
\begin{proof}
The statement follows from the facts that $(1)$ the region  $\vis(u,v)$ contains the upper part of $\R(v)$ and  $(2)$ the lower half-line delimiting the region $\vis(u,v)$ has a negative slope.
\end{proof}

\begin{figure}[htb]
\centering
\includegraphics[scale=0.7]{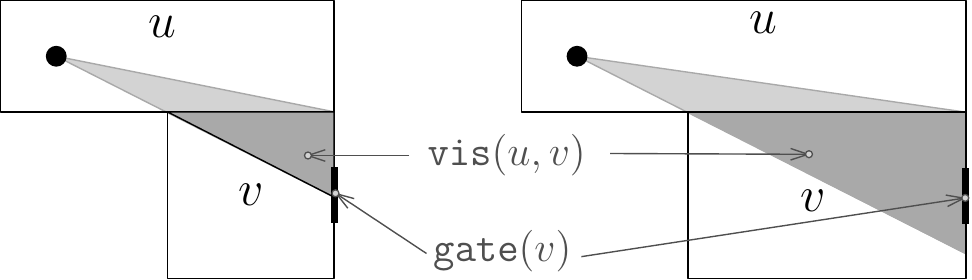}
\caption{Two horisontally adjacent rectangles such that $x_1(u)<x_1(v)$ and $x_2(u)=x_2(v)$, the case considered in Statement~\ref{stat:horizontal}. The fact $\gate(v) \subset \vis(u,v)$ does not hold before, but holds after the stretch of $u$ and $v$.  }
\label{fig:hor_adj}
\end{figure}

\begin{statement}
\label{stat:diverging}
Let $u$ and $v$ be two vertically adjacent rectangles, such that $u$ is to the left of $v$ and $u$ is a diverging neighbor of $v$.  There exists $X \geq x_2(v)$ such that, $\forall x\geq X$,  if we set $x_2(v) = x$  then the $\gate(v) \subset \vis(u,v)$. 
\end{statement}
\begin{proof}
The four cases determined by possible relations among the coordinates $y_1(u)$,  $y_2(u)$,  $y_1(v)$ and  $y_2(v)$ are shown on Figure~\ref{fig:diverg_adj}.
Since the region $\vis(u,v)$ is diverging, i.e. the half-lines through $u$ delimiting $\vis(u,v)$ have positive and negative slope, there exists $X\geq x_2(v)$, such that, for any $x\geq X$, if we set $x_2(v)=x$, then $\R(v) \subset \vis(u,v)$ and therefore $\gate(v) \subset \vis(u,v)$. 
\end{proof}

\begin{figure}[htb]
\centering
\includegraphics[scale=0.7]{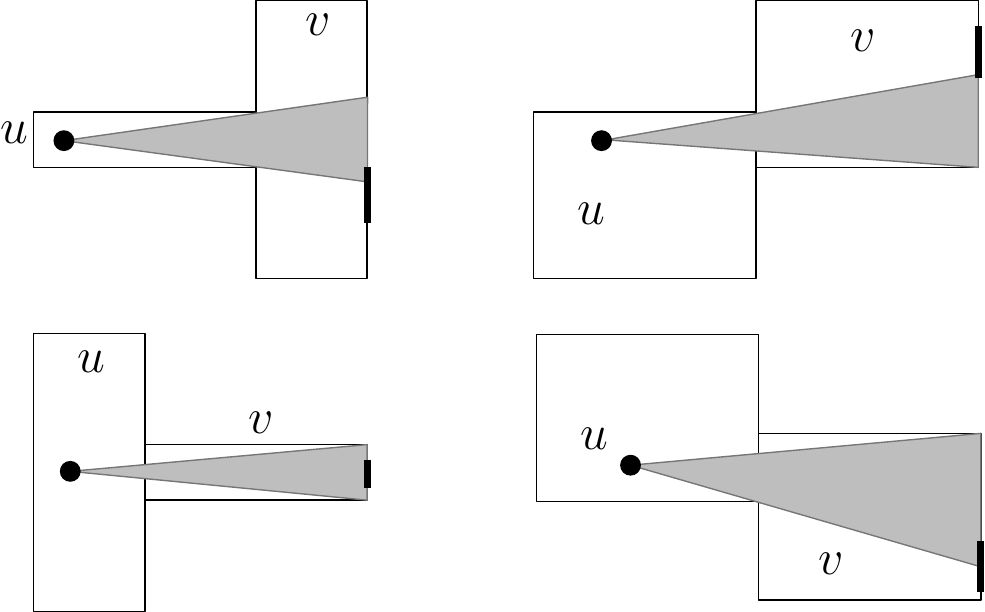}
\caption{The illustration of Statement~\ref{stat:diverging}. Four cases of vertical adjacency where $u$ is a diverging neighbor of $v$.}
\label{fig:diverg_adj}
\end{figure}

\begin{statement}
\label{stat:non-diverging-above}
Let $u$ and $v$ be two vertically adjacent rectangles, such that $u$ is to the left of $v$ and $y(u) > y_2[u,v]$ ($u$ is a non-diverging neighbor of $v$). If $\gate(v) \cap \blind(u,v) = \emptyset$ then  there exists $X \geq x_2(v)$ such that, $\forall x\geq X$,  if we set $x_2(v) = x$  then $\gate(v) \cap \blind(u,v) \neq \emptyset$.  Otherwise, there exists $y<y(u)$ such that if we set $y_2(v)=y$ then $\gate(v) \subset \vis(u,v)$.
\end{statement}
\begin{proof}
Assume first that $\gate(v) \cap \blind(u,v) = \emptyset$ (see Figure~\ref{fig:nondiverg_adj}, top left) for the illustration.
Since both lines delimiting $\vis(u,v)$ have negative slope, as $x_2(v)$ grows $gate(v)$ ``slides'' over $\vis(u,v)$ 
from bottom to top. Thus, there exists  $X \geq x_2(v)$ such that, $\forall x\geq X$,  if we set $x_2(v) = x$ then $\vis(u,v) \cap \blind(u,v) \neq \emptyset$ (Figure~\ref{fig:nondiverg_adj}, top right). 

Assume now that  $\gate(v) \cap \blind(u,v) \neq \emptyset$ (Figure~\ref{fig:nondiverg_adj}, bottom left). Observe that as $y_2(v)$ increases and remains less than $y(u)$, the slope of the topmost half-line delimiting $\vis(u,v)$ increases and remains negative. Thus, if   $y_2(v)=y(u)$, then the mentioned half-line has slope zero and $\blind(u,v)=\emptyset$. So, when $y_2(v)$ tends to $y(u)$, $\blind(u,v)$ tends to $\emptyset$. Recall that $\blind(u,v)$ is the topmost segment of $\R(v)$, and that $\gate(v)$ does not contain the topmost point of $R(v)$. Hence, there exists $y_2(v)<y<y(u)$ when $\blind(u,v)$ is small enough and does not intersect with $\gate(v)$ (Figure~\ref{fig:nondiverg_adj}, bottom right).  
\end{proof}

\begin{figure}[htb]
\centering
\includegraphics[scale=0.7]{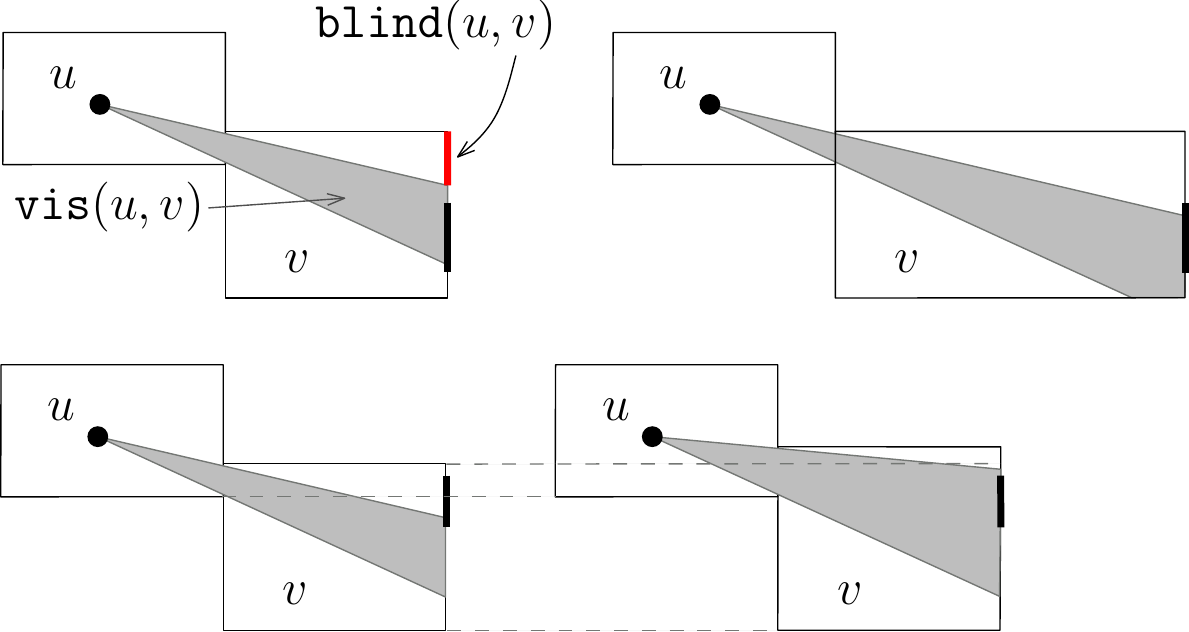}

\caption{The top left figure illustrates the case when $\gate(v) \cap \blind(u,v) = \emptyset$.  The top right figure shows that after stretching $v$ we get that $\gate(v) \cap \blind(u,v) \neq \emptyset$. The second line illustrates what happens when $y_2(v)$ is increased. In the left fugure we have that  $\gate(v) \cap \blind(u,v) \neq \emptyset$. The right figures shows the situation after the increase of $y_2(v)$, we get that $\gate(v) \subset \vis(u,v)$.}
\label{fig:nondiverg_adj}
\end{figure}

The  following statement is symmetric to Statement~\ref{stat:non-diverging-above} and can be proven identically.
\begin{statement}
\label{stat:non-diverging-below}
Let $u$ and $v$ be two vertically adjacent rectangles, $u$ is to the left of $v$ and $y(u) < y_1[u,v]$  ($u$ is a non-diverging neighbor of $v$). If $\gate(v) \cap \blind(u,v) = \emptyset$ then  there exists $X \geq x_2(v)$ such that, $\forall x\geq X$,  if we set $x_2(v) = x$  then the $\gate(v) \cap \blind(u,v) \neq \emptyset$.  Otherwise, there exists $y>y(u)$ such that if we set $y_2(v)=y$ then $\gate(v) \subset \vis(u,v)$.
\end{statement}


\section{Main result}
\label{sec:main}

\begin{theorem}
Let $G$ be a planar graph admitting a rectangular dual $D$. There exists a scaling $D'$ of $D$ such that $G$ and $D'$ admit a  straight-line simultaneous drawing.
\end{theorem}

\begin{proof}

We assume that $D$ is bounded by four rectangles (see Figure~\ref{fig:example_init}) that have horizontal adjacency between each other, if not so we add them to $D$, as well as the corresponding vertices to $G$.   For the simplicity of notation we denote the new graph by $G$ and its rectangular dual by $D$. After the scaling $D'$ of $D$ is created and the straight-line simultaneous drawing of $G$ and  $D'$ is constructed we simply remove the added vertices and rectangles. We denote the bottomost rectangle of $D$ by $v_S$, the topmost by $v_N$, the leftmost by $v_W$ and the rightmost by $v_E$.  
 
\begin{figure}[htb]
\centering
\includegraphics[scale=0.7]{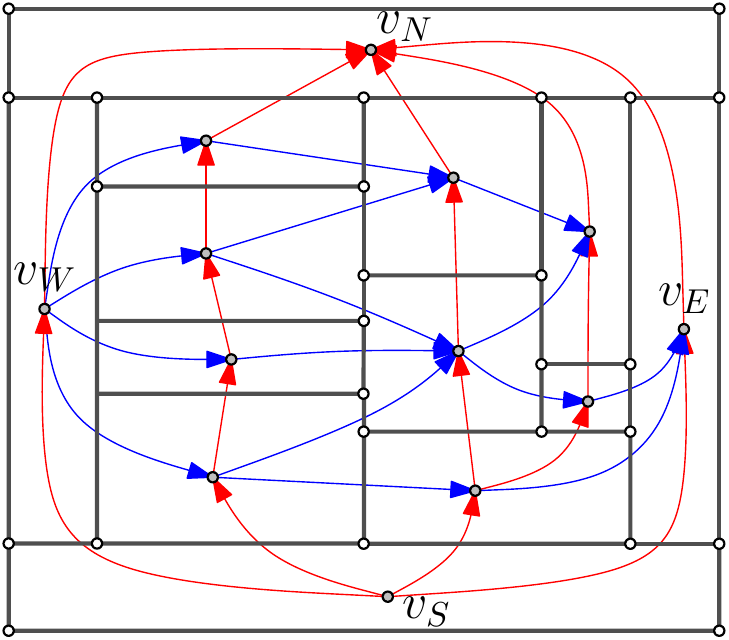}
\caption{Rectangular drawing $\Gamma$ of the primal graph $H$, which we treat as a contact of rectangles. The dual  $G$ of $H$ is colored according to its REL. This graph is used as an example throughout the paper.}
\label{fig:example_init}
\end{figure} 
 
Let  $E$ be the edge set of $G$ and let $E^R, E^B$ be the REL of $G$ that is defined by rectangular dual $D$ (refer to Figure~\ref{fig:example_init}).  Recall that $G^R$ is the subgraph of $G$ containing only the edges of $E^R$ plus the four external edges that are oriented so that $v_S$ is a source and $v_N$ is a sink. Recall also that   $G^R$ is an $st$-digraph.

\begin{figure}[htb]
\centering
\includegraphics[scale=0.5]{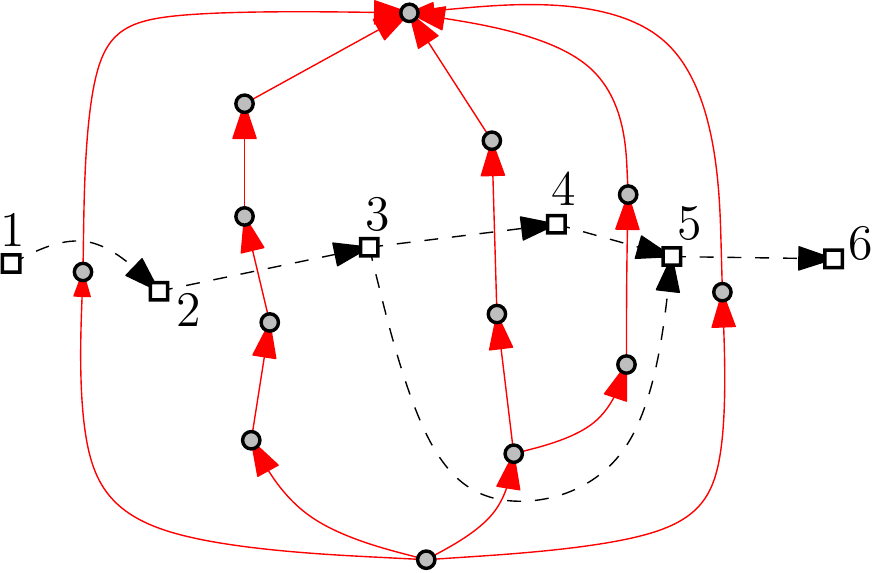}
\caption{ The graph $G^R$, its dual $(G^R)^\star$ and a topological ordering of $(G^R)^\star$, which in this case is unique. }
\label{fig:example_red}
\end{figure}

Let $(G^R)^\star$ be the  dual of $G^R$, it is an $st$-digraph by itself. Let $f_1,\dots,f_k$ be a topological ordering of the vertices of  $(G^R)^\star$ (see Figure~\ref{fig:example_red}). We denote by $G^R_i$, $1\leq i \leq k$ the subgraph of $G^R$ constituted by the vertices and edges of the faces $f_1,\dots,f_i$ (Figure~\ref{fig:example_subgraph}, left). As a special case, graph $G_1^R$ is the subgraph containing only vertices $v_S,~v_W, v_N$ and the edges $(v_S,v_W),~(v_W,v_N)$. While $G_k^R=G^R$.   It is not difficult to see that $G^R_i$ is an $st$-digraph and that $G^R_{i+1}$  can be constructed from $G^R_i$ by adding the right boundary of $f_{i+1}$ to the external face of $G^R_i$ (see Figure~\ref{fig:example_subgraph}). This fact is proven formally in~\cite[Lemma 4]{GiordanoLMSW15} for maximal planar $st$-digraphs, the proof for non-maximal planar $st$-digraph is along the same lines.   
Let $G_i$ be the subgraph of $G$ that is induced by the vertices of  
  $G^R_i$ (see Figure~\ref{fig:example_step}). Observe that the edges of $G_i$ that do not belong to $G_i^R$ are blue and lie in the internal faces of $G_i^R$.

\begin{figure}[htb]
\centering
\includegraphics[scale=0.6]{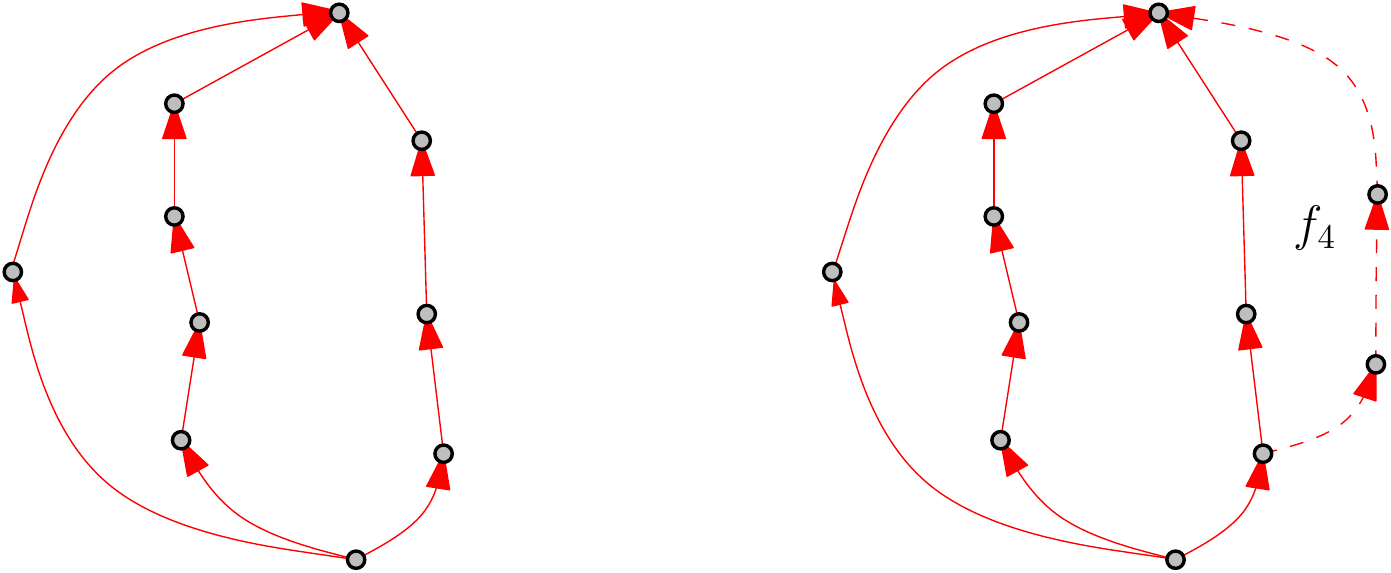}
\caption{Left: The graph $G^R_3$. Right: The graph $G^R_4$, which is produced from $G^R_3$ by adding the vertices and the edges of the right boundary of $f_4$ to the external face of $G^R_3$.}
\label{fig:example_subgraph}
\end{figure}

\begin{figure}[htb]
\centering
\includegraphics[scale=0.6]{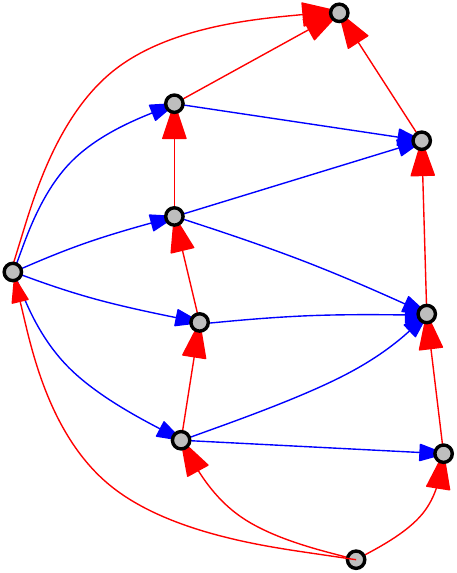}
\caption{The graph $G_3$.}
\label{fig:example_step}
\end{figure}

In the following, by induction on $i$, $1\leq i \leq k$, we construct a rectangular dual $D_i$ of $G_i$ 
consistent with the REL $E^R,~E^B$ (restricted to $G_i$), and show how to construct a straight-line simultaneous drawing of $G_i$ and $D_i$. For $i=k$ we will obtain a rectangular dual $D'$ of $G_k=G$ and the straight-line simultaneous drawing of $G$ and $D'$. Since $D'$ is consistent with the REL $E^R,~E^B$ it represents a scaling of $D$. So the theorem will follow.

 \begin{description}
\item[Base case: $i=1$.]   $G_1$  consists of the vertices $v_W$, $v_S$, $v_N$.  The graph $G_1$ contains only the red edges $(v_S,v_W)$  and $(v_W,v_N)$.  We represent the three vertices $v_S$, $v_W$, $v_N$ of $R_1$ as three rectangles such that $x_1(v_S)=x_1(v_W)=x_1(v_N)$, $x_2(v_S)=x_2(v_W)=x_2(v_N)$  and $y_2(v_S)=y_1(v_W)<y_2(v_W)=y_1(v_N)$ (see Figure~\ref{fig:example_base}). Graph $G_1$ contains exactly the same edges as $G_1^R$ and a straight-line simultaneous drawing of $G_1$ and $D_1$ can be constructed trivially.

\begin{figure}[htb]
\centering
\includegraphics[scale=0.9]{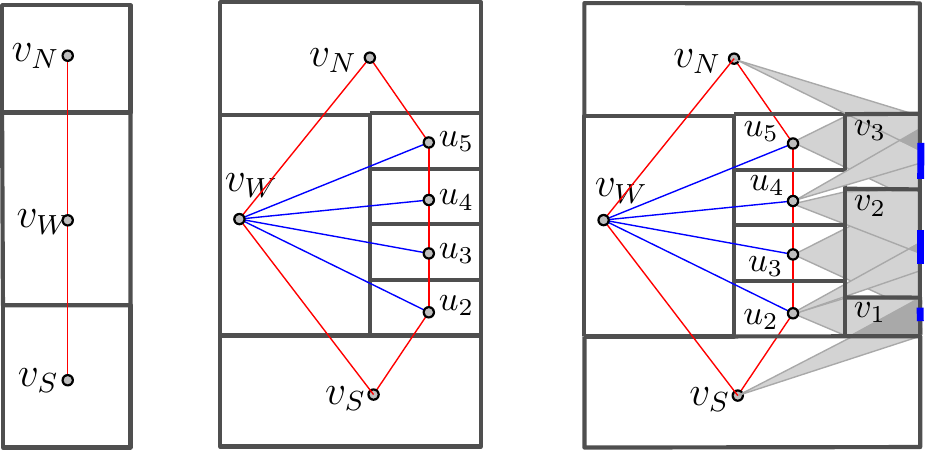}
\caption{Left: the base case. Middle and right: induction step illustrated on the graph of Figure~\ref{fig:example_init}. Before induction step we have  simultaneous drawing of $G_2$ and $D_2$ (middle). After induction step we have simultaneous drawing of $G_3$ and  $D_3$ (Figure~\ref{fig:example_indstep2}). Intermediate steps are also shown in Figure~\ref{fig:example_indstep}. Middle: rectangles $u_1,\dots,u_a$ appear on the right boundary of $D_i$. Right: placement the rectangles  $v_1,\dots,v_b$. The light grey regions show the  visibility of $u_1,\dots,u_a$ inside $v_1,\dots,v_b$. The dark grey regions show the intersection of visibility regions of all neighbors. }
\label{fig:example_base} 
\end{figure}

\item[Induction hypothesis.] For each $j\leq i<k$, there exists a rectangular dual $D_j$ of $G_j$ consistent with $(G^R, G^B)$ such that $G_j$ and $D_j$ have a straight-line simultaneous drawing.

\item[Induction step.] As mentioned above, $G^R_{i+1}$  can be constructed from $G^R_i$ by adding the right boundary of $f_{i+1}$ to the external face of $G^R_i$. Let directed path $u_1,\dots,u_a$ be the left boundary of $f_{i+1}$. Since vertices $u_1,\dots,u_a$ represent a directed sub-path of the right boundary of the external face of $G_i^R$, the rectangles $u_1,\dots,u_a$ appear on the right boundary of $D_i$ and lie on top of each other. See Figure~\ref{fig:example_base} for a specific example.  Let $u_1,v_1,\dots,v_b,u_a$ be the right boundary of $f_{i+1}$. We stretch all the rectangles of the right boundary of $D_i$ except for $u_2,\dots,u_{a-1}$ by the same value (Figure~\ref{fig:example_base}, right) and place the new rectangles  $v_1,\dots,v_b$ vertically between $u_1$ and $u_a$, according to the adjacency between the vertices $u_2,\dots,u_{a-1}$ and $v_1,\dots,v_b$. We set $x_2(v_1)=\dots=x_2(v_b) = x_2(u_a)$, thus the current $D_{i+1}$ is bounded by a rectangle and obviously comprises a rectangular dual of $G_i$, consistent with REL $E^R,~E^B$. 

Consider those of rectangles (resp. vertices) $v_1\dots,v_b$ that are adjacent to at least two rectangles (resp. vertices) among $u_1,\dots, u_a$, we call them \emph{critical}.   We next show that the rectangles of the right boundary of $D_{i+1}$ can be resized so that the gate of each critical rectangle is in the visibility region of each of its neighbor.  The resizing consists of two modifications; we first stretch the right boundary of $D_{i+1}$, which ensures visibility to the gates of some of the critical vertices, and fulfillment of a special condition for the gates of the remaining critical vertices. We then obtain visibility to the gates of these remaining critical vertices by moving vertically some common boundaries of $v_1,\dots,v_b$, such that the existing visibilities are preserved.  We do not perform any operation for non-critical vertices. Their placement is simple and will be explained at the end of the construction.

\begin{figure}[htb]
\centering
\includegraphics[scale=0.9]{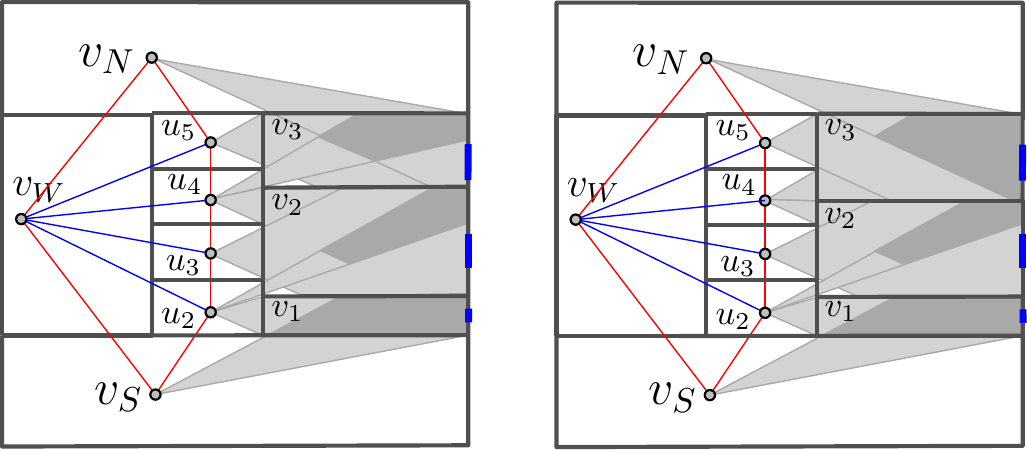}
\caption{\small Left:  After the horizontal stretch of the right boundary of $D_{3}$, where the initial $D_{3}$ is shown in Figure~\ref{fig:example_base}, right. The value of the stretch is determined by the pairs $(v_N,v_3),~(u_4,v_2)$ and $(u_2,v_2)$. For the pair $(v_N,v_3)$ (resp. $(u_4,v_2)$)  Statement~\ref{stat:horizontal} (resp. Statement~\ref{stat:diverging}) is applied, to achieve that $\gate(v_3) \subset \vis(v_N,v_3)$ (resp. $\gate(v_2) \subset \vis(u_4,v_2)$).  Finally, Statement~\ref{stat:non-diverging-below} applies for the pair $(u_2,v_2)$, to archive that $\gate(v_2) \cap \blind(u_2,v_2) \neq \emptyset$.  The remaining adjacencies does not increase the value of the stretch. Right: Statement~\ref{stat:non-diverging-below} is applied to the pair $(u_4,v_3)$ to archive that $\gate(v_3) \subset \vis(u_4,v_3)$  and as a result the common boundary $[v_2,v_3]$ is moved down.}
\label{fig:example_indstep} 
\end{figure}

We first further specify the positions of the gates of those of $v_1,\dots,v_b$ which are critical. See Figure~\ref{fig:gate_placement}, left for the illustration.  Consider a critical vertex $v_q$, $1\leq q \leq b$, and let  $p$ be the minimum index and $\ell$ be the maximum index such that $1\leq p \leq \ell \leq a$ and $u_p$ and $u_\ell$  are adjacent to $v_q$. Positions of the vertices $u_p$ and $u_\ell$ are known by induction hypothesis. We place $\gate(v_q)$ so that $y(u_p)< y_1(\gate(v_q))<y_2(\gate(v_q))<y(u_\ell)$. The reason for this positioning of the gates for the critical vertices will become clear later in the proof.  The gates of the non-critical vertices are not of any interest to us, since they will not be used.

\begin{figure}[htb]
\centering
\includegraphics[scale=0.9]{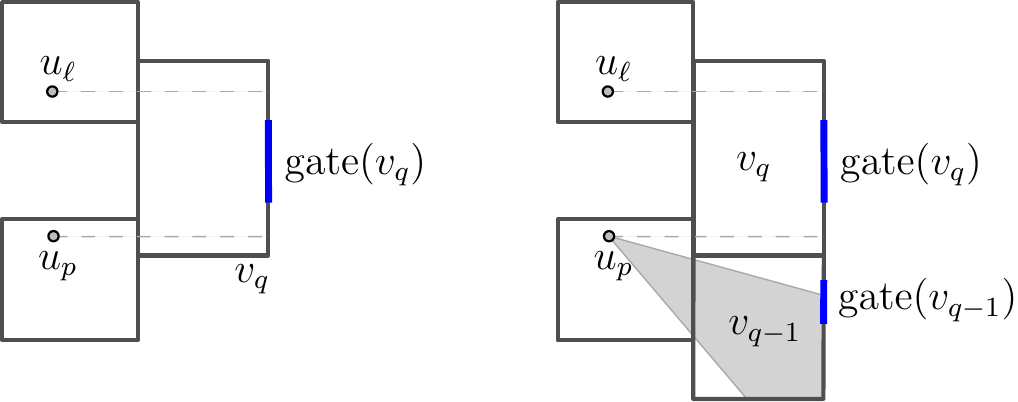}
\caption{\small Left: Illustration for the placement of gates of the critical vertices. Right: After second modification it holds that $\gate(v_q) \subset R(v_q)$. }
\label{fig:gate_placement} 
\end{figure}

Consider a $v_q$, $1 \leq q \leq b$, and its neighbor, say $u_p$, $1\leq p \leq a$. If $u_p$ is a diverging neighbor of $v_q$ then  Statement~\ref{stat:diverging} applies and determines the value $X(u_p,v_q)\geq x_2(v_q)$,  such that  $\forall x\geq X(u_p,v_q)$,  if we set $x_2(v) = x$  then $\gate(v_q) \subset \vis(u_p,v_q)$.  If $u_p$ is a non-diverging neighbor of $v_q$, but $\gate(v_q) \cap \blind(u_p,v_q) = \emptyset$ then Statement~\ref{stat:non-diverging-above} or Statement~\ref{stat:non-diverging-below} determine the value $X(u_p,v_q)\geq x_2(v_q)$, such that $\forall x\geq X(u_p,v_q)$,  if we set $x_2(v) = x$  then $\gate(v_q) \cap \blind(u_p,v_q) \neq \emptyset$.   Finally, the values $X(u_1,v_1)$ and $X(u_a,v_b)$ are determined by Statement~\ref{stat:horizontal}.
 Let $X = \max\{X(u_p,v_q) | 1\leq p \leq a, 1\leq q \leq b\}$.

\begin{figure}[htb]
\centering
\includegraphics[scale=0.9]{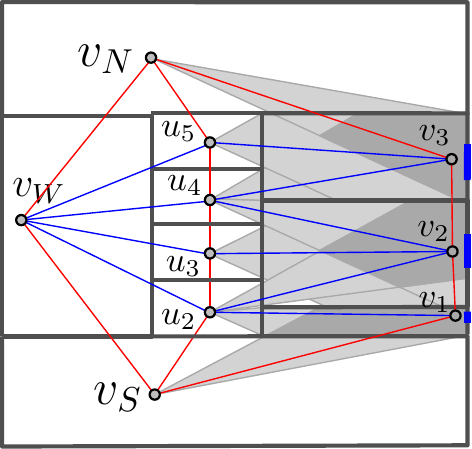}
\caption{\small  Statement~\ref{stat:non-diverging-below} is applied to the pair $(u_2,v_2)$ to archive that $\gate(v_2) \subset \vis(u_2,v_2)$ and as a result the common boundary $[v_1,v_2]$ is moved down.}
\label{fig:example_indstep2} 
\end{figure}

We now stretch the right boundary of $D_{i+1}$ to have coordinate $X$ (refer to Figure~\ref{fig:example_indstep}, left). By Statement~\ref{stat:diverging}, if $u_p$, $1 \leq p \leq a$ , is a diverging neighbor of $v_q$, $1 \leq q \leq b$ ,then  $\gate(v_q) \subset \vis(u_p,v_q)$.

\begin{figure}[htb]
\centering
\includegraphics[scale=0.9]{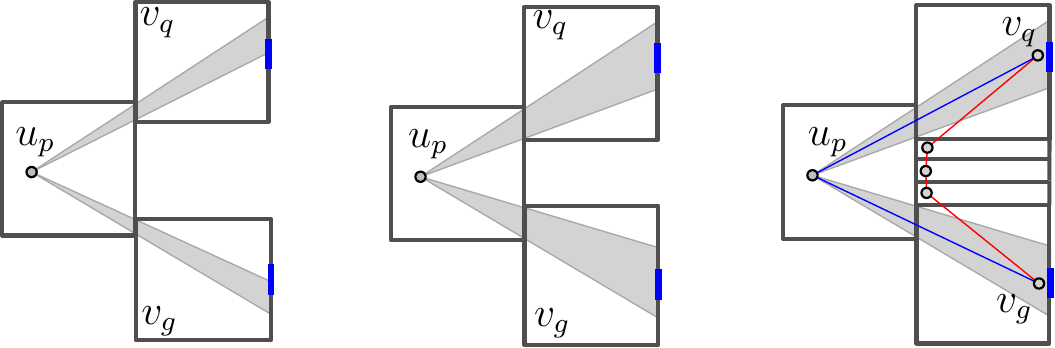}
\caption{\small  Illustration of the second  modification, where the common horizontal boundaries of $v_1,\dots,v_b$ are possibly moved. Left:  Vertex  $u_p$ is a non-diverging neighbour of $v_q$ and $v_g$,  it holds that $\blind(u_p,v_q) \cap \gate(v_q) \neq \emptyset$ and $\blind(u_p,v_g) \cap \gate(v_g) \neq \emptyset$.  Middle: Statement~\ref{stat:non-diverging-above} is applied to $u_p$ and $v_q$. Statement~\ref{stat:non-diverging-below} is applied to $u_p$ and $v_g$. Right: the non-critical neighbors of $u_p$ are placed and the vertices of $G$ are positioned as explained in the proof. }
\label{fig:nondiverg_neigh} 
\end{figure}

Let now $v_q$, $1\leq q \leq b$, be a vertex with non-diverging neighbor $u_p$, $1\leq p \leq a$. The previous modification ensures that $\gate(v_q) \cap \blind(u_p,v_q) \neq \emptyset$.   By Statements~\ref{stat:non-diverging-above}~and~\ref{stat:non-diverging-below},  the exists $Y(u_p,v_q)<y(u_p)$ (case $y(u_p)>y_2[u_p,v_q]$) or $Y(u_p,v_q) > y(u_p)$ (case $y(u_p) < y_1[u_p,v_q]$) such that if we set $y_2(v_q)=Y(u_p,v_q)$ then $\gate(v_q) \subset \vis(u_p,v_q)$ (refer to Figure~\ref{fig:example_indstep}, right and to Figure~\ref{fig:example_indstep2} for the specific example, refer also to Figure~\ref{fig:nondiverg_neigh} for an abstract example). 

In the following we show that the second modification does now destroy the visibilities to the critical vertices which existed after the first modification.  First, observe that after the second modification the  $\gate(v_q)$, $1\leq q \leq b$ of a critical vertex $v_q$  still belongs to $R(v_q)$.  This is ensured by the initial placement of the gates of the critical vertices. Thus, consider Figure~\ref{fig:gate_placement} (right), where $u_p$ is a non-diverging neighbor of $v_{q-1}$. The result of the second modification will be that the common boundary $[v_q,v_{q-1}]$ is moved up to  $y<y(u_p)$. By the placement of the $\gate(v_q)$ we have that $y_1(\gate(v_q)) > y(u_p)$ and therefore $y_1(\gate(v_q))>y$. By a symmetric argument for $v_q$, $v_{q+1}$ and $u_\ell$ we infer that  $\gate(v_q) \subset R(v_q)$. 

Second, assume  that a vertex $u_p$ is a non-diverging neighbor of $v_q$ such that $y(u_p)<y_1[u_p,v_q]$ and such that $\gate(v_q) \cap \blind(u_p,v_q) \neq \emptyset$ (Figure~\ref{fig:nondiverg_neigh}), then the application of Statement~\ref{stat:non-diverging-above} results in moving the segment $[v_q,v_{q-1}]$  down to a $y$-coordinate $Y(u_p,v_q)>y(u_p)$. The last inequality ensures that the visibility of $u_p$ inside $v_{q-1}$ has not changed.

We now explain how to draw the non-critical vertices, consider again Figure~\ref{fig:nondiverg_neigh}. Assume that $u_p$ is a non-diverging neighbour of  $v_q$, such that $y(u_p)<y_1[u_p,v_q]$  and  $v_g$, such that $y(u_p)>y_1[u_p,v_g]$.  As already mentioned,  the application of Statement~\ref{stat:non-diverging-above} results in moving the segment $[v_q,v_{q-1}]$  down to a $y$-coordinate $Y(u_p,v_q)>y(u_p)$. The application of Statement~\ref{stat:non-diverging-below} results in moving $[v_g,v_{g+1}]$ up to $Y(u_p,v_g) < y(u_p)$. Thus, $y[v_g,v_{g+1}] < y(u_p) < y[v_{p-1},v_p]$ and the space between $y[v_g,v_{g+1}]$ and $y[v_{p-1},v_p]$ is used  for the placement of the non-critical neighbors of $u_p$ (see Figure~\ref{fig:nondiverg_neigh}, right). It may happen that $v_{g+1}=v_p$, then we set $y_2(v_g)=y_1(v_q)=y(u_p)$.

Next, we place the actual vertices of the right boundary of $f_{i+1}$, except of $u_1$ and $u_a$, which have already been placed by induction hypothesis. For each critical vertex $v_q$ we place vertex $v_q$ very close to the middle of $\gate(v_q)$. Since for each neighbor $u_p$ of $v_q$  $\gate(v_q) \subset \vis(u_p,v_q)$, the straight-line edge $(u_p,v_q)$ crosses only $[u_p,v_q]$. A non-critical vertex $v_q$, which is adjacent to a single vertex $u_p$, is placed arbitrarily close to the common segment $[u_p,v_q]$.
Thus, the straight-line edge $(u_p,v_q)$ crosses only $[u_p,v_q]$. Edges $(u_1,v_1)$ and $(v_b,u_a)$ cross the segments $[u_1,v_1]$ and $[v_b,u_a]$, respectively, as ensured by the application of Statement~\ref{stat:horizontal}.
 Finally, each edge $(v_j,v_{j+1})$, $1\leq j \leq b-1$, crosses only $[v_j,v_{j+1}]$, since $x_1(v_j) = x_1(v_{j+1})$ and $x_2(v_j) = x_2(v_{j+1})$.   This concludes the proof of the theorem. \qed

\end{description}  
\end{proof}


\section{Conclusion}
In this paper we  considered the problem of drawing simultaneously a planar graph and its rectangular dual. We required that the vertices of the primal are positioned in the corresponding rectangles, the drawing of the primal graph is planar and straight-line, and each edge of the primal crosses only the rectangles where its end-points lie. Our proof in constructive and leads to a linear-time algorithm. However, the vertices  are not placed on the grid and the area requirements of the construction are unclear. It would be interesting to either refine the algorithm to produce a simultaneous drawing with polynomial area, or to construct a counterexample, requiring an exponential area. 

\section*{Acknowledgments} I would like to thank Md. Jawaherul Alam,  Michael Kaufmann, Stephen Kobourov and Roman Prutkin for the useful discussions during the preliminary stage of this work.

\bibliography{bibliography}
\bibliographystyle{plain}

\end{document}